\documentclass[10pt]{amsart}
\usepackage{amsmath,amsthm}
\usepackage[colorlinks=true,linkcolor=blue,urlcolor=blue,citecolor=blue]{hyperref}
\usepackage[all]{xy}

\newtheorem{theorem}{Theorem}[section]
\newtheorem{lemma}[theorem]{Lemma}

\theoremstyle{definition}
\newtheorem{remark}[theorem]{Remark}

\def\L2{L^{2}}

\def\M{\mathcal{M}}
\def\E{\mathcal{E}}
\def\B{\mathcal{B}}
\def\D{\mathcal{D}}

\def\R{\Bbb{R}}

\def\m1{^{-1}}

\def\H{\mathcal{H}}
\def\F{\mathcal{F}}

\def\R{\mathbb{R}}
\def\o{\omega}

\def\O{\Omega}

\def\Q{\mathcal{Q}}

\begin{document}

\title[On the eigenvalue counting function...]{On the eigenvalue counting function for Schr\"odinger operator: some upper bounds.}%
\author{Fabio Cipriani}%
\address{Dipartimento di Matematica, Politecnico di Milano, piazza Leonardo da Vinci 32, 20133 Milano, Italy.}%
\email{fabio.cipriani@polimi.it}%
\thanks{}

\date{December 10, 2018}
\dedicatory{}
\footnote{This work has been supported by Laboratoire Ypatia des Sciences Mathématiques C.N.R.S. France - Laboratorio Ypatia delle Scienze Matematiche I.N.D.A.M. Italy (LYSM)}
\begin{abstract}
The aim of this work is to provide an upper bound on the eigenvalues counting function $N(\R^n,-\Delta+V,e)$ of a Sch\"odinger operator $-\Delta +V$ on $\R^n$ corresponding to a potential $V\in L^{\frac{n}{2}+\varepsilon}(\R^n,dx)$, in terms of the sum of the eigenvalues counting function of the Dirichlet integral $\D$ with Dirichlet boundary conditions on the subpotential domain $\{V< e\}$, endowed with weighted Lebesgue measure $(V-e)_-\cdot dx$ and the eigenvalues counting function of the absorption-to-reflection operator on the equipotential surface $\{V=e\}$.
\end{abstract}
\maketitle
\section{Introduction and description of the main results}
To describe the content of the present work, we recall an iconic result of H. Weyl [18] concerning a problem posed by the physicist H.A. Lorenz and stimulated by problems arising in J. Jeans' radiation theory, about the asymptotic distribution of the eigenvalues $0<\lambda_1\le\cdots\le \lambda_k\le\cdots $ (repeated according their multiplicity) of the Laplace operator $-\Delta$ subject to Dirichlet conditions on the boundary $\partial\O$ of a bounded domain $\O\subset\R^n$:
\[
\lambda_k\sim C_n\cdot |\O|^{-2/n}\cdot k^{2/n}\qquad k\to +\infty\, .
\]
Hence, from the spectrum of $-\Delta$ geometric information can be extracted such as the volume $|\O|$ of the region. If $N(\O,\Delta,\mu)$ denote the number of eigenvalues, counted according their multiplicity, not exceeding the value $\mu>0$, then the Weyl's result follows from the estimate
\[
N(\O,-\Delta,\mu)\sim C_n\cdot |\O|\cdot \mu^{n/2}\qquad\mu\to +\infty
\]
just noticing that $N(\O,-\Delta,\lambda_k)=k$ for any integer $k\ge 1$. \\
Here $C_n:=(4\pi)^{-n/2}\Gamma(1+n/2)^{-1}=(2\pi)^{-n/2}\cdot \o_n$ is the so called classical constant, $\o_n$ being the volume of the unit ball in $\R^n$.
\par\noindent
G. P\'olya [15] proved that for domains tiling $\R^n$, the following equivalent one-side bounds
\begin{equation}
\lambda_k\ge C_n\cdot |\O|^{-2/n}\cdot k^{2/n}\, ,\quad k\ge 1\, ,\qquad N(\O,-\Delta,\mu)\le C_n\cdot |\O|\cdot \mu^{n/2}\, , \qquad \mu\ge 0
\end{equation}
and conjectured that these are true for all bounded domains. In this perspective, E.H. Lieb [13] proved the above inequalities where the classical constant $C_n$ is replaced by a strictly greater one $L_n>C_n$. Later, P. Li-S.T. Yau [11] obtained inequalities with the constants $\frac{2\pi n}{{\rm e}}$ and $\frac{nC_n}{n+2}$ which are worst than Lieb's ones but that both agree with the Weyl's asymptotic result in the sense that $C_n\sim L_n\sim \frac{2\pi n}{\rm e}\sim \frac{nC_n}{n+2}$.\\
While the works of G. P\'olya [15] were motivated by problems arising in continuous mechanics and in particular those of vibrating membranes, those of E.H. Lieb were motivated by problems in Quantum Mechanics. More specifically, by the problem to bound above the number $N(\R^n,-\Delta + V,\mu)$ of eigenvalues of a Schr\"odinger operator on $\R^n$ associated to a potential $V$. Lieb obtained, for potentials $V\in L^{n/2}(\R^n,dx)$ on $\R^n$ with $n\ge 3$, the upper bound
\begin{equation}
N(\R^n,-\Delta + V,\mu)\le L_n\cdot \int_{\R^n}(V-\mu)_-^{n/2}\cdot dx\qquad \mu\in\R
\end{equation}
from which, among other things, his one sided bound on $N(\O,-\Delta,\mu)$ follows. The bound for the eigenvalues of the Schr\"odinger operator are subtler than those for the Laplace operator. For example, the former are definitely not true in low dimension $n=1,2$. The bound (1.2) is referred as {\it semiclassical} because the integral appearing in (1.2) is proportional by $\o_n/n$ to the volume of the region $\{(p,q)\in\R^n\times\R^n:|p|^2+V(q)\le\mu\}$ in the classical phase space. The semiclassical bound was obtained independently (with different method) and published almost simultaneously by M. Cwikel and G.V. Rosenbljum  (with constants worst than $L_n$) and it is often referred as the Cwikel-Lieb-Rosenbljum bound (see [13] for details). In particular M. Cwikel exploited ideas introduced by B. Simon [16] who previously proved an inequality of the form
\begin{equation}
N(\R^n,-\Delta + V,0)\le S_{n,\varepsilon}\cdot\Bigl(\|V_-\|_{n/2+\varepsilon} + \|V_+\|_{n/2+\varepsilon}\Bigr)^{n/2}\, .
\end{equation}
for potentials $V\in L^{n/2+\varepsilon}(\R^n,dx)$, with $S_{n,\varepsilon}\to+\infty$ as $\varepsilon\to 0$. The method followed by [13] is based on a reduction argument leading to a Birman-Schwinger compact operator [3], [17] followed by a Wiener integral representation of its trace.
\vskip0.2truecm
In this work the method we follow to bound above $N(\O,-\Delta+V,e)$ for $\O\subseteq\R^n$ with $n\ge 3$, is based not directly on considerations of self-adjoint, semibounded operators but rather on properties of their corresponding quadratic forms, often Dirichlet forms.
\vskip0.1truecm\noindent
In Section 2 we reduce the problem to bound above $N(\O,-\Delta+V,e)$ to the one to bound above the number $N(\D,H^1(U_e,m_e),1)$ of eigenvalues not exceeding the level $1$ of the operator corresponding to the Dirichlet integral $\D$ on the space $L^2(U_e,m_e)$ where $U_e:=\{V<e\}$ is the sublevel set of the potential $V$ and the background reference measure $m_e:=(V-e)_-\cdot dx$ is the Lebesgue one weighted by the potential.
\vskip0.1truecm\noindent
In Section 3 we study a family of quadratic forms $(\E_\lambda,\F)$ on the boundary space $L^2(\partial U_e,\mu_e)$ where $\partial U_e:=\{V=e\}$ is the level set of the potential $V$ and $\mu_e$ is the measure on $\partial U_e$ obtained averaging by $m_e$ the family of harmonic measures of $\partial U_e$. These forms are defined as traces, in the Sobolev or Dirichlet forms sense, of quadratic forms on $H^1(U_e,m_e)$ associated to the subspaces of $\lambda$-harmonic functions of finite energy.\\
Forms in this family are termed {\it absorption-to-reflection} quadratic forms to suggest that they are generalization of those associated to the Dirichlet-to-Neumann operators of smooth Euclidean domains [1], [2], [8].\\
In particular we show that $\E_\lambda$ is bounded below by the Dirichlet form $\E_0$ up to a constant multiple, depending on $\lambda$, of $\|\cdot\|^2_{L^2(\partial U_e,\mu_e)}$.
\vskip0.1truecm\noindent
In Section 4 we first prove, for $\lambda\ge 0$ in the resolvent set of $(\D_{U_e},H^1_0(U_e,m_e))$, the splitting
\[
N((\D_{U_e},H^1(U_e,m_e)),\lambda) =N((\D_{U_e},H^1_0(U_e,m_e)),\lambda) + N((\E_\lambda,\F),0)
\]
in terms of the counting function of the Dirichlet integral on the weighted Sobolev subspace $H^1_0(U_e,m_e)$ corresponding to Dirichlet boundary conditions on $\partial U_e$ plus the number of nonpositive eigenvalues of the absorption-to reflection quadratic form $(\E_\lambda,\F)$.\\
The splitting above generalizes the one obtained by L. Friedlander [9] in the proof of the Payne conjecture [14] about Dirichlet and Neumann eigenvalues of Euclideans domains.
\vskip0.1truecm\noindent
Subsequently, in the same Section 4, we show that $N((\E_\lambda,\F),0)$ is bounded above by the eigenvalues counting number $N((\E_0,\F),\|\lambda\cdot A_\lambda\|)$. Here $A_\lambda:=-\Delta (-\Delta-\lambda)^{-1}$ where $-\Delta$ is the operator whose quadratic form is the Dirichlet form with Dirichlet boundary conditions $(\D,H^1_0(U_e,m_e))$.
\vskip0.1truecm\noindent
The final Section 5 is devoted to obtain Weyl upper bounds on the counting functions above.
In Section 5.1 we obtain the upper bound for $\lambda\ge 0$
\[
N((\D_{U_e},H^1_0(U_e,m_e)),\lambda)\le {\rm e}^{2d}S_n^d\cdot \|(V-e)_-\|_{L^1(U_e,dx)}^2\cdot \|(V-e)_-\|^{d-2}_{L^p(U_e,dx)}\cdot\lambda^{d}
\]
for a suitable effective dimension $d$ depending upon $p$.\\
In Section 5.2 we obtain, for a suitable effective dimension $m$, the upper bound
\[
N((\E_0,\F),\gamma)\le {\rm e}^{2m}\cdot \|(V-e)_-\|^2_{L^1(U_e,dx)}\cdot(c_1\gamma +c_2)^m\qquad \gamma\ge 0\, ,
\]
assuming that $\partial U_e$ is smooth. Here the coefficients $c_1\, ,c_2$ depend upon some $L^p$-norms of the Radon-Nikodym derivative of the boundary measure $\mu_e$ with respect to the Hausdorff $(n-1)$-dimensional measure of $\partial U_e$.
\vskip0.1truecm\noindent
In Section 5, the method to bound above $N((\D_{U_e},H^1_0(U_e,m_e))$ and $N((\E_0,\F),\gamma)$ is essentially the same: we start from the classical Sobolev inequalities on $U_e$ or from the Sobolev trace inequalities in the $\partial U_e$ case, then we prove Sobolev inequalities with respect to the measures $m_e$ on $U_e$ or $\mu_e$ on $\partial U_e$. Then we use the Davies-Simon [7], [6] theory of ultracontractivity to convert these informations into uniform boundedness of heat kernels, then into bounds on the trace of the corresponding Markov semigroups and finally into bounds on the eigenvalues counting functions.
\vskip0.5truecm\noindent
{\bf Warning}: in the rest of the work an italic style letter {\it "e"} will continue to mean a fixed level of the potential function $V$, while a roman style letter "e" will represent the Neper number.

\section{Schr\"odinger and Dirichlet energy integrals and comparison of their eigenvalues counting functions}
In the following, when $(\E,\F)$ is a lower semibounded, closed quadratic form on a Hilbert space $\H$, we shall denote by $N((\E,\F),\beta)$ the number of  eigenvalues, counted according to their multiplicity, of the corresponding lower semibounded, self-adjoint operator $(L,D(L))$ on $\H$ which do not exceed the value $\beta\in\R$. In other words, denoting by $E^L$ the spectral measure of $(L,D(L))$, we define
\[
N((\E,\F),\beta):={\rm Tr}(E^L({(-\infty,\beta]}))
\]
as the trace of the spectral projection corresponding to the interval $(-\infty,\beta]$.
\vskip0.2truecm\noindent
We shall denote by $dx$ the Lebesgue measure of $\R^n$ and by ${\rm BL}(\O)$ the space of Beppo Levi functions (see [4])
\[
{\rm BL}(\O):=\{u\in L^2_{\rm loc}(\O,dx):|\nabla u|\in L^2(\O,dx)\}\, .
\]
Whenever $\O\subseteq\R^n$ is an open set endowed with positive Radon measure $m$, $(\D_\O,H^1(\O,m))$ will denote the Dirichlet integral
\[
\D_\O[u]:=\int_\O |\nabla u(x)|^2\cdot dx
\]
defined on the space $H^1(\O,m):={\rm BL}(\O)\cap L^2(\O,m)$.\\
We shall denote by $(\D_\O,H^1_0(\O,m))$ the Dirichlet integral considered on the subspace $H^1_0(\O,m)$ obtained as the closure of $H^1(\O,m)\cap C(\O)$ in the graph norm of $H^1(\O,m)$.\\
When $m=dx$ the form $(\D_\O,H^1(\O,dx))$ is closed on $L^2(\O,dx)$. Moreover, $C^\infty_c(\O)$ is a form core for $(\D_D,H^1_0(\O,dx))$ and the corresponding nonnegative, self-adjoint operator is the Laplacian $-\Delta$ subject to Dirichlet boundary conditions.
\vskip0.2truecm\noindent
Assume $V$ to be a {\it negative, upper semicontinuous potential in the Kato class}
\[
V=-V_-\in K_n(\O,dx)
\]
together with the (self-adjoint, lower semibounded) Schr\"odinger operator $-\Delta+V$ whose (closed, symmetric, lower semibounded) quadratic form is given by
\[
\Q_V[u]:=\D_\O[u]+\int_\O |u|^2 V\, dx\qquad u\in H^1_0(\O,dx)\, .
\]
For the background material on Schr\"odinger operators we refer to [5]. Fix a nonpositive energy level $e\le 0$ and consider the open sublevel set
\[
U_e:=\{x\in\O: V(x)< e\}
\]
of the potential energy, endowed with the weighted Lebesgue measure $m_e(dx):=(V-e)_-\, dx$. Since $V_-$ is assumed to lies in the Kato class, $H^1_0(\O,dx)$ can be considered as a subspace of $H^1(U_e,m_e)$.
\begin{lemma}
The Dirichlet integral
\[
\D_{U_e}[u]:=\int_{U_e}|\nabla u|^2\, dx\qquad u\in H^1(U_e,m_e)
\]
is a Dirichlet form on $L^2(U_e,m_e)$.
\end{lemma}
\begin{proof}
Since the form is clearly Markovian, we have just to prove that it is closed. Suppose that $u_n\in H^1(U_e,m_e)$ is a $\D_{U_e}$-Cauchy sequence converging to some $u\in L^2(U_e,m_e)$ in the norm of $L^2(U_e,m_e)$. Then, possibly passing to a subsequence, we have that, $m_e$-a.e. on $U_e$, $u_n\to u$. Since $m_e$ and $dx$ are equivalent on $U_e$, we have also that, $dx$-a.e. on $U_e$, $u_n\to u$. Since $H^1(u_e,m_e)\subset {\rm BL}(U_e)$, by the properties of the $\D_{U_e}$-convergence in ${\rm BL}(U_e)$, there exists a sequence of constants $c_n$ and $v\in {\rm BL}(U_e)$ such that $\D[u_n -v]\to 0$ and $u_n +c_n\to v$ in $L^2_{\rm loc}(U_e,m_e)$. Then, possibly passing to a subsequence, we have that, $dx$-a.e. on $U_e$, $u_n+c_n\to v$. Hence, $c_n=(u_n +c_n)-u_n\to v-u$, $dx$-a.e. on $U_e$. On the other hand, the limit of a sequence of constant which converges $dx$-a.e. on $U_e$ can only be a constant function $c$ on $U_e$, so that $c=v-u$, $dx$-a.e. on $U_e$. Since $\D[u_n-u]=\D[u_n-v +c])=\D[u_n -v]\to 0$, we have that $u_n$ converges to $u$ in the form norm of $H^1(U_e,m_e)$.
\end{proof}
The following observation, appearing in [11 Corollary 2], is a reformulation of the {\it reduction argument} of Birman and Schwinger which was also employed by [13]. \\
While the Birman-Schwinger reduction identifies the number $N((\Q_V,H^1_0(\O,dx)),e)$ of eigenvalues of the Schr\"odinger operator
$-\Delta-V_-$ on $L^2(\O,dx)$, not exceeding the value $e\le 0$, with the number of eigenvalues of the Birman-Schwinger compact operator (associated to the Birman-Schwinger kernel) greater or equal to $1$, the following elementary observation compares $N((\Q_V,H^1_0(\O,dx)),e)$ with the number $N((\D_{U_e},H^1(U_e,m_e)),1)$ of eigenvalues of the Dirichlet integral $(\D_{U_e},H^1(U_e,m_e))$ on $L^2(U_e,m_e)$, not exceeding the value $1$.
\begin{lemma}
For all $\lambda \ge 1$, we then have
\[
N((\Q_V,H^1_0(\O,dx)),e) \le N((\D_{U_e},H^1(U_e,m_e)),\lambda)\, .
\]
\end{lemma}
\begin{proof}
Since for all $\lambda\ge 1$ and all $u\in H^1_0(\O,dx)$ we have
\[
\begin{split}
\Q_{V}[u]-e\|u\|^2_{L^2(\O,dx)}&=\D_\O[u]+\int_{\O}|u|^2\cdot (V-e)\, dx \\
&=\D_\O[u]+\int_{\R^n}|u|^2\cdot (V-e)_+\, dx - \int_{\O}|u|^2\cdot (V-e)_-\, dx \\
&\ge \D_{U_e}[u] - \int_{U_e}|u|^2\, dm_e \\
&\ge \D_{U_e}[u] - \lambda\int_{U_e}|u|^2\, dm_e\, ,
\end{split}
\]
the subspace of $H^1_0(\O,dx)$ where the quadratic form $\Q_V$ is bounded by $e$ with respect to the norm of $L^2(\O,m)$ is contained in (or it can be identified by restriction with) the subspace of $H^1(U_e,m_e)$ where the Dirichlet integral $\D_{U_e}$ is bounded by $1$ with respect to the norm of $L^2(U_e,\mu_e)$. The result then follows by the Min-Max Theorem.
\end{proof}
\begin{remark}
The above result can be restated saying that the number of bound states of a quantum particle subject to a potential $V$, whose energy does not exceed the level $e\in \R$, is less or equal the number of bound states of energy not exceeding the level $1$ of a free particle moving in a background where the reference measure $m_e$ is the Lebesgue one weighted by the potential $(V-e)_-$.\\
It can be considered as a quantum version of the Jacobi trick by which the orbits of a classical particle moving under the influence of a potential $V$ are geodesics of the Jacobi (conformally equivalent) metric.
\end{remark}
\section{absorption-to-Reflection quadratic forms and operators}
The goal of the present section is to compare, in a natural way, the eigenvalues distribution of the Dirichlet integral $D_{U_e}$ when considered on the space $H^1(U_e,m_e)$ to the eigenvalues distribution of the Dirichlet integral $D_{U_e}$ when considered the space $H^1_0(U_e,m_e)$, through the eigenvalues distributions of a family of operators on the boundary $\partial U_e$. These operators, which from the point of view of Dirichlet forms theory may be called {\it absorption-to-reflection} operators, generalize the Dirichlet-to-Neumann operators on the boundary $\partial\O$ of smooth Euclidean domains $\O$, well studied in literature (see [1], [2], [8]). The difference lies in the fact that instead of starting from the Sobolev space $H^1(\O,dx)$ and its subspace $H^1_0(\O,dx)$ we start from $H^1(\O,m)$ and $H^1_0(\O,m)$, for a positive Radon measure $m$ on $\O$ and that the {\it absorption-to-reflection} operators on the  boundary are closed with respect to a measure on $\partial\O$ depending on $m$ and no more with respect to the Hausdorff $(n-1)$-dimensional measure.
\vskip0.2truecm\noindent
For $\lambda\in\mathbb{R}$ let us consider the space of {\it finite energy, $\lambda$-harmonic functions}
\[
\mathcal{H}_\lambda:=\{u\in H^1(U_e,m_e):\D_{U_e}(v|u)-\lambda(v|u)_{L^2(U_e,m_e)}=0,\,\, v\in H^1_0(U_e,m_e)\}\, .
\]
\begin{lemma}
Let us consider the quadratic form $(\D_{U_e},H^1_0(U_e,m_e))$ on $L^2(U_e,m_e)$. Then for any value $\lambda\notin \sigma(\D_{U_e},H^1_0(U_e,m_e))$, the following direct splitting holds true
\[
H^1(U_e,m_e)=H^1_0(U_e,m_e)\oplus \mathcal{H}_\lambda\, .
\]
\end{lemma}
\noindent
Recall that the extended Dirichlet space $H^1(U_e,m_e)_e$ is the space of measurable functions on $U_e$ which are $m_e$-a.e. pointwise limits of
$\D_{U_e}$-Cauchy sequences of the Dirichlet space $H^1(U_e,m_e)$. If $V$ is continuous and $m_e$ is finite, then $H^1(U_e,m_e)_e$ reduces to the space ${\rm BL}(U_e)$ of Beppo Levi functions ([4 Theorem 2.2.14]).
\par\noindent
Let us consider the level set $\partial U_e=\{V=e\}\subset\O$ of the potential energy, the function space
\[
\F_{\partial U_e}:=\{u|\partial U_e:u\in H^1(U_e,m_e)_e\}
\]
and the trace operator
\[
{\rm Tr}:H^1(U_e,m_e)\to \F_{\partial U_e}\qquad {\rm Tr}(u):=u|\partial U_e\, .
\]
Since $ker({\rm Tr})=H^1_0(U_e,m_e)$, the previous splitting provides the following
\begin{lemma}
For any $\lambda\notin \sigma(\D_{U_e},H^1_0(U_e,m_e))$ we have
\[
{\rm Tr}(\mathcal{H}_\lambda)={\rm Tr}(H^1(U_e,m_e))
\]
and the trace operator is a linear isomorphism between $\mathcal{H}_\lambda$ and $\F:={\rm Tr}(H^1(U_e,m_e))$.
\end{lemma}
We introduce now the {\it weak solution operator} $L_\lambda$ of the Dirichlet problem associated to a Dirichlet space $H^1(U_e,m_e)$ on
$L^2(U_e,m_e)$.
\begin{lemma}
For $\lambda\notin \sigma(\D_{U_e},H^1_0(U_e,m_e))$ a linear operator $L_\lambda:\F\to \mathcal{H}_\lambda$ is defined assigning to $\varphi\in \F$ the unique $L_\lambda \varphi\in \mathcal{H}_\lambda$ such that ${\rm Tr}(L_\lambda\varphi)=\varphi$.\\
The function $L_\lambda\varphi\in H^1(U_e,m_e)$ is the unique minimizer of the quadratic functional
\[
\mathcal{L}_\lambda :H^1(U_e,m_e)\to [0,+\infty)\qquad \mathcal{L}_\lambda [u]:=\D_{u_e}[u]-\lambda \|u\|^2_{L^2(U_e,m_e)}
\]
on the set $C_\varphi :=\{u\in H^1(U_e,m_e):{\rm Tr}(u)=\varphi\}$.
\end{lemma}
\begin{proof}
The Dirichlet integral $\D_{U_e}$, the norm square $\|\cdot\|^2_{L^2(U_e,m_e)}$ and the functional $\mathcal{L}_\lambda$ are continuous functional on the Dirichlet space $H^1(U_e,m_e)$ endowed with the Hilbertian norm $\bigl(\D_{U_e}[\cdot]+\|\cdot\|^2_{L^2(U_e,m_e)}\bigr)^{1/2}$. Since $C_\varphi\subset H^1(U_e,m_e)$ is a closed and convex set, the existence and uniqueness follows from the projection theorem on closed convex sets in Hilbert spaces.
\end{proof}
Recall that we denote by $-\Delta$ the self-adjoint, nonnegative operator on $L^2(U_e,m_e)$ whose closed quadratic form is the Dirichlet form
$(\D_{U_e},H^1_0(U_e,m_e))$.
\par\noindent
Next results expresses the fact that the operators $L_\lambda$ can be expressed through the one corresponding to the value $\lambda =0$ by a bounded operator which is functional calculus of $-\Delta$.
\begin{lemma}
Assuming  $0\notin \sigma(\D_{U_e},H^1_0(U_e,m_e))$ and for any $\lambda\notin \sigma(\D_{U_e},H^1_0(U_e,m_e))$, the operator
$A_\lambda :=-\Delta(-\Delta-\lambda)^{-1}$ is self-adjoint and bounded on $L^2(U_e,m_e)$ and on $H^1(U_e,m_e)$ and establishes a continuous isomorphism between the spaces
$\mathcal{H}_0$ and $\mathcal{H}_\lambda$ such that
\[
L_\lambda =A_\lambda\circ L_0\, .
\]
\end{lemma}
\begin{proof}
The boundedness of $A_\lambda$ on $L^2(U_e,m_e)$ follows from the Spectral Theorem by the assumption that $\lambda$ belongs to the resolvent set. The boundedness of $A_\lambda$ on $H^1(U_e,m_e)$ follows from
\[
\begin{split}
\|A_\lambda u\|^2_{H^1(U_e,m_e)}&=\|(-\Delta +I)^{1/2}(-\Delta)(-\Delta-\lambda)^{-1}u\|^2_2 \\
&=\|(-\Delta)(-\Delta-\lambda)^{-1}(-\Delta +I)^{1/2}u\|^2_2 \\
&=\|A_\lambda(-\Delta +I)^{1/2}u\|^2_2 \\
&\le\|A_\lambda\|^2_{L^2\to L^2}\cdot \|(-\Delta +I)^{1/2}u\|^2_2 \\
&=\|A_\lambda\|^2_{L^2\to L^2}\cdot \|u\|^2_{H^1(U_e,m_e)} \qquad u\in H^1(U_e,m_e)
\end{split}
\]
so that, in particular, $\|A_\lambda\|^2_{H^1\to H^1}     \le       \|A_\lambda\|^2_{L^2 \to L^2}$. The invertibility of $A_\lambda$ follows from the assumptions $0,\lambda\notin \sigma(\D_{U_e},H^1_0(U_e,m_e))$ and one may check that $A_\lambda^{-1}=I-\lambda (-\Delta)^{-1}$.
\par\noindent
Let $\varphi\in \F$ so that, by definition, $L_0\varphi\in \mathcal{H}_0, L_\lambda\varphi\in \mathcal{H}_\lambda$ and ${\rm Tr}(L_0\varphi)={\rm Tr}(L_\lambda\varphi)=\varphi$. Setting $v:=L_\lambda\varphi-L_0\varphi\in H^1(U_e,m_e)$, since ${\rm Tr}(v)=\varphi-\varphi=0$, we have $v\in H^1_0(U_e,m_e)$, so that
\[\begin{split}
\D_{U_e}(w|L_\lambda\varphi)&=\lambda(w|L_\lambda\varphi)_2 \\
\D_{U_e}(w|L_0\varphi)&=0 \\
\D_{U_e}(w|v)&=\lambda(w|L_\lambda\varphi)_2\qquad\qquad w\in H^1_0(U_e,m_e)\, .
\end{split}
\]
Hence $L_\lambda\varphi-L_0\varphi=v=(-\Delta)^{-1}(\lambda L_\lambda\varphi)$ so that $L_\lambda\varphi=(-\Delta)(-\Delta-\lambda)^{-1}L_0\varphi=A_\lambda L_0\varphi$.
\end{proof}
For our present purposes, it is convenient to restate the above result as a relation between quadratic forms.
\begin{lemma}
For any $\lambda\notin \sigma(\D_{U_e},H^1_0(U_e,m_e))$, consider the quadratic form $(\E_\lambda,\F)$ defined as
\[
\E_\lambda[\varphi]:=\D_{U_e}[L_\lambda\varphi]-\lambda\cdot \|L_\lambda\varphi\|^2_{L^2(U_e,m_e)}\qquad \varphi\in \F\, .
\]
Then, if $0\notin \sigma(\D_{U_e},H^1_0(U_e,m_e))$ we have
\[
\E_0[\varphi]-\E_\lambda[\varphi]=\lambda\cdot (L_0\varphi | A_\lambda L_0\varphi)_{L^2(U_e,m_e)}\qquad \varphi\in \F\, .
\]
\end{lemma}
\begin{proof}
Since $L_\lambda\varphi\in\mathcal{H}_\lambda$ and $L_\lambda\varphi_0-L_0\varphi\in H^1_0(U_e,m_e)$ we have
\[
\D_{U_e}(L_0\varphi |L_\lambda\varphi_0-L_0\varphi)=0\, ,\qquad \D_{U_e}(L_\lambda\varphi_0-L_0\varphi|L_\lambda\varphi)=\lambda (L_\lambda\varphi_0-L_0\varphi|L_\lambda\varphi)\, .
\]
Then
\[
\begin{split}
\E_\lambda[\varphi]&=\D_{U_e}[L_\lambda\varphi]-\lambda\cdot \|L_\lambda\varphi\|^2_{L^2(U_e,m_e)} \\
&=\D_{U_e}[L_0\varphi +(L_\lambda\varphi-L_0\varphi)]-\lambda\cdot \|L_\lambda\varphi\|^2_{L^2(U_e,m_e)} \\
&=\D_{U_e}[L_0\varphi]+\D_{U_e}[L_\lambda\varphi-L_0\varphi] +2{\rm Re}\bigl(\D_{U_e}(L_0\varphi|L_\lambda\varphi-L_0\varphi)\bigr)-\lambda\cdot \|L_\lambda\varphi\|^2_{L^2(U_e,m_e)} \\
&=\E_0[\varphi]+\D_{U_e}(L_\lambda\varphi-L_0\varphi |L_\lambda\varphi)-\D_{U_e}(L_\lambda\varphi-L_0\varphi |L_0\varphi)-\lambda\cdot \|L_\lambda\varphi\|^2_{L^2(U_e,m_e)} \\
&=\E_0[\varphi]+\D_{U_e}(L_\lambda\varphi-L_0\varphi |L_\lambda\varphi)-\lambda\cdot \|L_\lambda\varphi\|^2_{L^2(U_e,m_e)} \\
&=\E_0[\varphi]+\lambda\cdot(L_\lambda\varphi-L_0\varphi |L_\lambda\varphi)_{L^2(U_e,m_e)}-\lambda\cdot \|L_\lambda\varphi\|^2_{L^2(U_e,m_e)} \\
&=\E_0[\varphi]-\lambda\cdot(L_0\varphi |L_\lambda\varphi)_{L^2(U_e,m_e)} \\
&=\E_0[\varphi]-\lambda\cdot(L_0\varphi |A_\lambda L_0\varphi)_{L^2(U_e,m_e)}\, .
\end{split}
\]
\end{proof}
Let us recall that the Dirichlet space $(\D_{U_e},H^1(U_e,m_e))$ on $L^2(\bar{U_e},{\bf 1}_{{\bar{U_e}}}\cdot dx)$ is {\it regular} in the sense of Dirichlet form theory (see [4]) if the involutive subalgebra $H^1(U_e,m_e)\cap C_0( \bar{U_e})$ is a form core, uniformly dense in $C_0(\bar{U_e})$. This is the case, for example, if $U_e$ has {\it continuous boundary} in the sense of Maz'ya (see [4]) and in particular if the potential $V$ is continuous.
\begin{lemma}
If the Dirichlet space $(\D_{U_e},H^1(U_e,m_e))$ on $L^2(\bar{U_e},{\bf 1}_{{\bar{U_e}}}\cdot dx)$ is regular, then
\vskip0.2truecm\noindent
i) the algebra $\B:=\F\cap C_0(\partial U_e)$ is uniformly dense in $C_0(\partial U_e)$\\
ii) the map $L_0:\B\to \mathcal{H}_0$ extends to a Markovian map from $C_0(\partial U_e)$ to $C_b(U_e)$ and\\
iii) for any fixed $x\in U_e$ there exists a unique probability measure $\mu_x$ on $\partial U_e$ such that
\[
(L_0\varphi)(x)=\int_{\partial U_e}\varphi\, d\mu_x\qquad \varphi\in \B\, .
\]
\end{lemma}
\begin{proof}
i) The result follows from $H^1(U_e,m_e)\cap C_0( \bar{U_e})\subseteq \B$ and the regularity assumption.
ii) By the Maximum Principle for harmonic functions, $0\le L_0\varphi\le 1$ for all $\varphi\in \B$ such that $0\le\varphi\le 1$. The map $L_0$ is then continuous w.r.t. the uniform norm and, since $\B$ is norm dense in $C_0(\partial U_e)$, it extends to a Markovian map from $C_0(\partial U_e)$ to $C_b(\partial U_e)$. iii) The functional $\varphi\mapsto (L_0\varphi)(x)$ is then positive on $C_0(\partial U_e)$ and it can be represented by a positive  measure $\mu_x$ on $\partial U_e$. Since $L_01=1$ we have $\mu_x(\partial U_e)=(L_0 1)(x)=1$.
\end{proof}\noindent
The measures $\{\mu_x:x\in U_e\}$ on $\partial U_e$ are the {\it harmonic measures} of the Euclidean domain $U_e$. In particular, like the operator $L_0$, they are independent upon the measure $m_e$ and a fortiori upon the potential.\\
Next results show that the measure $m_e=(V-e)_-\cdot dx$ on $U_e$ and the family of harmonic measures provide a natural measure on the boundary
$\partial U_e$ with respect to which the difference between the quadratic forms above can be conveniently considered.
\begin{lemma}
Assume the Dirichlet space $(\D_{U_e},H^1(U_e,m_e))$ on $L^2(\bar{U_e},{\bf 1}_{{\bar{U_e}}}\cdot dx)$ to be regular and consider the positive measure $\mu_e$ on $\partial U_e$ defined by
\[
\mu_e:=\int_{U_e} m_e(dx)\,\mu_x\, .
\]
Then, under the assumption $0\notin \sigma(\D_{U_e},H^1_0(U_e,m_e))$ and for any $\lambda\notin \sigma(\D_{U_e},H^1_0(U_e,m_e))$, the quadratic forms $\E_0$ and $\E_\lambda$ differ by a bounded quadratic form on the Hilbert space $L^2(\partial U_e,\mu_e)$. In particular we have
\[
\Bigl|\E_0[\varphi]-\E_\lambda[\varphi]\Bigr|\le \|\lambda\cdot A_\lambda\|_{L^2(U_e,m_e)\to L^2(U_e,m_e)}\cdot \|\varphi\|^2_{L^2(\partial U_e,\mu_e)}\qquad \varphi\in \F\, .
\]
\end{lemma}
\begin{proof}
By H\"older inequality, for all $\varphi\in\B$ we have
\[
\begin{split}
\|L_0\varphi\|^2_{L^2(\partial U_e)}&=\int_{U_e}|(L_0\varphi)(x)|^2\, m_e(dx)=\int_{U_e}m_e(dx)\Bigl|\int_{\partial U_e}\varphi (y)\, \mu_x(dy)\Bigr|^2 \\
&\le \int_{U_e}m_e(dx)\int_{\partial U_e}|\varphi (y)|^2\, \mu_x(dy) \\
&=\|\varphi\|^2_{L^2(\partial U_e,\mu_e)}\, .
\end{split}
\]
Thus $L_0$ extends to a contraction on from $L^2(\partial U_e,\mu_e)$ to $L^2(U_e,m_e)$ and, by Lemma 3.5 above, we have for $\varphi\in \F$
\[
\begin{split}
\Bigl|\E_0[\varphi]-\E_\lambda[\varphi]\Bigr|&=(L_0\varphi | \lambda\cdot A_\lambda L_0\varphi)_{L^2(U_e,m_e)}| \\
&\le \|\lambda\cdot A_\lambda\|_{L^2(U_e,m_e)\to L^2(U_e,m_e)}\cdot \|\varphi\|^2_{L^2(\partial U_e,\mu_e)}\, .
\end{split}
\]
\end{proof}
\begin{lemma}
Assume the Dirichlet space $(\D_{U_e},H^1(U_e,m_e))$ on $L^2(\bar{U_e},{\bf 1}_{{\bar{U_e}}}\cdot dx)$ to be regular and $0\notin \sigma(\D_{U_e},H^1_0(U_e,m_e))$. Then, for any $\lambda\notin \sigma(\D_{U_e},H^1_0(U_e,m_e))$, the quadratic forms $(\E_0,\F)$ and $(\E_\lambda,\F)$ are closed on $L^2(\partial U_e,\mu_e)$. In particular, the former is nonnegative, the latter is lower semibounded and the following bound holds true
\[
\E_0[\varphi]-\|\lambda A_\lambda\|\cdot \|\varphi\|^2_{L^2(\partial U_e,\mu_e)}\le\E_\lambda[\varphi]\qquad \varphi \in \F\, .
\]
Finally, $(\E_0,\F)$ is a Dirichlet form on $L^2(\partial U_e,\mu_e)$.
\end{lemma}
\begin{remark}
For $\lambda\notin \sigma(\D_{U_e},H^1_0(U_e,m_e))$, the self-adjoint operator $B_\lambda$ on the Hilbert space $L^2(\partial U_e,\mu_e)$ whose quadratic form is $(\E_\lambda,\F)$ will be called the {\it $\lambda$-absobtion-to-reflection} operator of the Dirichlet space $(\D_{U_e},H^1(U_e,m_e))$. In particular, if $0\notin \sigma(\D_{U_e},H^1_0(U_e,m_e))$, the nonnegative, self-adjoint operator $B_0$ will be called the {\it absorption-to-reflection} operator of the Dirichlet space $(\D_{U_e},H^1(U_e,m_e))$.
\end{remark}
\section{Comparison of eigenvalues counting functions}
The following is the main result of the work.\\
Let us denote by  $N((\E_\lambda,\F),0)$ the number of nonpositive eigenvalues of the quadratic form $(\E_\lambda,\F)$ of the {\it $\lambda$-absorption-to-reflection} operator $B_\lambda$ on the Hilbert space $L^2(\partial U_e,\mu_e)$.\vskip0.2truecm\noindent
Notice that the set of $\lambda\ge 1$ for which $\lambda\notin \sigma(\D_{U_e},H^1_0(U_e,m_e))$ is not empty if, for example,
$\sigma(\D_{U_e},H^1_0(U_e,m_e))$ is discrete.
\begin{theorem}
Assume the Dirichlet space $(\D_{U_e},H^1(U_e,m_e))$ on $L^2(\bar{U_e},{\bf 1}_{{\bar{U_e}}}\cdot dx)$ to be regular. Then, for any $\lambda\notin \sigma(\D_{U_e},H^1_0(U_e,m_e))$ we have
\begin{equation}
N((\D_{U_e},H^1(U_e,m_e)),\lambda) =N((\D_{U_e},H^1_0(U_e,m_e)),\lambda) + N((\E_\lambda,\F),0)
\end{equation}
and
\begin{equation}
N((\D_{U_e},H^1(U_e,m_e)),\lambda) \le N((\D_{U_e},H^1_0(U_e,m_e)),\lambda) + N((\E_0,\F),\|\lambda\cdot A_\lambda\|)\, ,
\end{equation}
If moreover $\lambda\ge 1$, we have
\begin{equation}
\begin{split}
N((\Q_V,H^1_0(\O,dx)),e)
&\le N((\D_{U_e},H^1_0(U_e,m_e)),\lambda) + N((\E_0,\F),\lambda\|A_\lambda\|)\, .
\end{split}
\end{equation}
\end{theorem}
\begin{proof}
Consider the closed, quadratic form
\[
\mathcal{L}_\lambda:H^1(U_e,m_e)\to [0,+\infty)\qquad \mathcal{L}_\lambda[u]:=\D_{U_e}[u]-\lambda\cdot \|u\|_{L^2(U_e,m_e)}
\]
and notice that the direct splitting $H^1(U_e,m_e)=H^1_0(U_e,m_e)\oplus \H_\lambda$ is $\mathcal{L}_\lambda$-orthogonal
\[
\mathcal{L}_\lambda[u_0 + u^\lambda]=\mathcal{L}_\lambda[u_0]+\mathcal{L}_\lambda[u^\lambda]\qquad u_0\oplus u^\lambda\in H^1_0(U_e,m_e)\oplus \H_\lambda\, .
\]
Let $M\subset H^1(U_e,m_e)$ (resp. $M_0\subset H^1_0(U_e,m_e)$, $\mathcal{M}_\lambda\subset\H_\lambda$) be the subspace where $\mathcal{L}_\lambda$ is negative on $H^1(U_e,m_e)$ (resp. $H^1_0(U_e,m_e)$, $\H_\lambda$) so that $N((\D_{U_e},H^1(U_e,m_e)),\lambda)={\rm dim}(M)$ and $N((\D_{U_e},H^1_0(U_e,m_e)),\lambda)={\rm dim}(M_0)$. If we consider $M_0^\prime:=M\cap H^1_0(U_e,m_e)$ and $\mathcal{M}_\lambda^\prime :=M\cap \mathcal{M}_\lambda$, we have $M_0^\prime\subseteq M_0$ and $\mathcal{M}_\lambda^\prime\subseteq \mathcal{M}_\lambda$ so that ${\rm dim}(M)={\rm dim}(M_0^\prime)+{\rm dim}(\mathcal{M}_\lambda^\prime)\le {\rm dim}(M_0)+{\rm dim}(\mathcal{M}_\lambda)$. On the other hand, since $M_0$ and $\mathcal{M}_\lambda$ are $\mathcal{L}_\lambda$-orthogonal, we have $M_0\oplus \mathcal{M}_\lambda\subseteq M$ so that ${\rm dim}(M_0)+{\rm dim}(\mathcal{M}_\lambda)\le {\rm dim}(M)$ and
\[
{\rm dim}(M)-{\rm dim}(M_0)={\rm dim}(\mathcal{M}_\lambda)\, .
\]
Notice now that the quadratic form $(\mathcal{L}_\lambda,\H_\lambda)$ considered on the Hilbert space $L^2(U_e,m_e)$ is isomorphic under the Markovian map $L_0:L^2(\partial U_e,\mu_e)\to L^2(U_e,m_e)$ to the quadratic form $(\E_\lambda,\F)$ on the Hilbert space $L^2(\partial U_e,\mu_e)$
\[
\E_\lambda[\varphi]:=\mathcal{L}_\lambda[L_0\varphi]\qquad \F:=\F_{\partial U_e}\cap L^2(\partial U_e,\mu_e)\, .
\]
Since the quadratic form $(\E_\lambda,\F)$ on the Hilbert space $L^2(\partial U_e,\mu_e)$ is closed, we have ${\rm dim}(\mathcal{M}_\lambda)=N((\E_\lambda,\F),0)$ so that ${\rm dim}(M)-{\rm dim}(M_0)=N((\E_\lambda,\F),0)$. The other bounds follow from previous lemma.
\end{proof}

\section{Weyl's type bounds on eigenvalues counting functions}
In the following sections we provide Weyl's type bounds on the eigenvalues counting function of the Dirichlet forms $(\D,H^1_0(U_e,m_e))$ on $L^2(U_e,m_e)$ and $(\E_0,\F)$ on $L^2(\partial U_e,\mu_e)$.
\subsection{}
In this section we show how to bound the first term $N((\D_{U_e},H^1_0(U_e,m_e)),\lambda)$ in the above evaluation of $N((\D_{U_e},H^1(U_e,m_e)),\lambda)$.\\
Under an hypothesis of $L^p$-integrability of $(V-e)_-$, for some $p>n/2$ and for $n\ge 3$, we prove initially a Sobolev inequality for the Dirichlet form
$(\D_{U_e},H^1_0(U_e,m_e))$ on the weighted Lebesgue space $L^2(U_e,m_e)$.\\
Then, using the Davies-Simon theory [7], [6], we turn these into a family of logarithmic Sobolev inequalities to prove that the Markov semigroup associated to $(\D_{U_e},H^1_0(U_e,m_e))$ is ultracontractive on $L^2(U_e,m_e)$.\\
Finally, assuming $(V-e)_-$ to be integrable, we show that the Markovian semigroup is nuclear so that the spectrum of its generator is discrete and that an upper bound of Weyl's type holds true on $N((\D_{U_e},H^1_0(U_e,m_e)),\lambda)$.
\vskip0.2truecm\noindent
Recall the Sobolev inequality for the Euclidean domain $U_e\subseteq\R^n$, $n\ge 3$,
\[
\|u\|^2_{L^{n^*}(U_e,dx)}\le S_n\cdot \D_{U_e}[u]\qquad u\in H^1_0(U_e,dx)\, ,
\]
where the best constant is given by
\[
S_n:=\frac{1}{n(n-2)\pi}\Bigl(\frac{\Gamma(n)}{\Gamma(n/2)}\Bigr)^{2/n}\, .
\]
\begin{lemma}
Suppose $(V-e)_-\in L^p(\O,dx)$ for some $p>n/2$ and set $n^*:=\frac{2n}{(n-2)}$, $n\ge 3$. Then the following weighted Sobolev inequality holds true \[
\|u\|^2_{L^r(U_e,m_e)}\le S_r(U_e,m_e)\cdot \D_{U_e}[u]\qquad u\in H^1_0(U_e,m_e)
\]
for $r:=n^*(1-p^{-1})>2$ and the Sobolev constant $S_r(U_e,m_e):= S_n\cdot \|(V-e)_-\|^{2/r}_{L^p(\O,dx)}$.
\end{lemma}
\begin{proof}
Setting $q:=(1-p^{-1})^{-1}$ we have $rq=n^*$ so that
\[
\begin{split}
\|u\|^2_{L^r(U_e,m_e)}&=\Bigl(\int_{U_e}|u|^r (V-e)_-\, dx\Bigr)^{2/r} \le \Bigl(\int_{U_e}|u|^{rq}\, dx\Bigr)^{2/rq}\cdot  \Bigl(\int_{U_e}(V-e)_-^p\, dx\Bigr)^{2/rp}\\
&=\|u\|^2_{L^{n^*}(U_e,dx)}\cdot \|(V-e)_-\|^{2/r}_{L^p(\O,dx)} \\
&\le S_n\cdot \|(V-e)_-\|^{2/r}_{L^p(\O,dx)} \cdot \D_{U_e}[u]\qquad u\in H^1_0(U_e,m_e)\, .
\end{split}
\]
\end{proof}
Let $L_e$ be the nonnegative, self-adjoint operator on $L^2(U_e,m_e)$ whose quadratic form is the Dirichlet form $(D_{U_e},H^1_0(U_e,m_e))$.
\begin{lemma}
Suppose $(V-e)_-\in L^p(\O,dx)$, for some $p>n/2$, $n\ge 3$. The Markovian semigroup $e^{-tL_e}$ on $L^2(U_e,m_e)$ is then ultracontractive
\[
\|e^{-tL_e}\|_{L^2\to L^\infty}\le \Bigl({\rm e}(d/4)S_r(U_e,m_e)\Bigr)^{d/4}\cdot t^{-d/4}\qquad t>0
\]
and its heat kernel is bounded by
\[
e^{-tL_e}(x,y)\le c\cdot t^{-d/2}\qquad m_e-a.e.\quad x,y\in U_e\, ,
\]
where $d:=\frac{2r}{r-2}>n$ and $c:=\Bigl({\rm e}\frac{d}{2}S_r(U_e,m_e)\Bigr)^{d/2}$.
\end{lemma}
\begin{proof}
Ultracontractivity follows from the weighted Sobolev inequality applying [6 Thm 2.4.2]
\[
\|e^{-tL_e}\|_{L^2\to L^\infty}\le c\cdot t^{-\frac{2r}{r-2}}\qquad t> 0\, .
\]
To evaluate explicitly the constant $c>0$, notice that, following the proof of [6 Thm 2.4.2], the Sobolev inequality implies the logarithmic Sobolev inequalities
\[
\int_{U_e}|u|^2\ln |u|\, dm_e \le (d/4)\cdot\Bigl(-\ln\varepsilon +\varepsilon\cdot S_r(U_e,m_e)\cdot\D_{U_e}[u]\Bigr)\qquad \varepsilon >0
\]
for all norm one functions $u\in H^1_0(U_e,m_e)$ and $d:=\frac{2r}{r-2}$. By rescaling, these appear as
\[
\int_{U_e}|u|^2\ln |u|\, dm_e \le \varepsilon\cdot\D_{U_e}[u]+\beta(\varepsilon)\qquad \varepsilon >0
\]
for $\beta(\varepsilon):=(d/4)\ln (rS_r(U_e,m_e)/4)-(d/4)\ln \varepsilon$. By [6 Corollary 2.2.8] we have
\[
\|e^{-tL_e}\|_{L^2\to L^\infty}\le e^{\frac{1}{t}\int_0^t\beta(\varepsilon)\, d\varepsilon} = \Bigl({\rm e}(d/4)S_r(U_e,m_e)\Bigr)^{d/4}\cdot t^{-d/4}\qquad t>0\, .
\]
By [6 Lemma 2.1.2] we then have the stated uniform upper bound on the heat kernel.
\end{proof}
\begin{lemma}
Suppose $(V-e)_-\in L^1(\O,dx)\cap L^p(\O,dx)$ for some $p>n/2$, $n\ge 3$. Then
\begin{itemize}
\item the spectrum of the Dirichlet form  $(D_{U_e},H^1_0(U_e,m_e))$ on $L^2(U_e,m_e)$ is discrete,
\item the associated Markovian semigroup is nuclear and
\[
{\rm Tr}(e^{-tL_e})\le \|(V-e)_-\|_{L^1(\O,dx)}^2\cdot \Bigl(edS_r(U_e,m_e)\Bigr)^{d}\cdot t^{-d}\qquad t>0.
\]
\end{itemize}
\end{lemma}
\begin{proof}
Since $m_e(U_e)=\int_{U_e}(V-e)_-\,dx =\|(V-e)_-\|_{L^1(\O,dx)}<+\infty$, by Lemma 2.12 and [6 Thm 2.1.4], the spectrum of the Dirichlet form  $(D_{U_e},H^1_0(U_e,m_e))$ on $L^2(U_e,m_e)$ is discrete and
\[
{\rm Tr}(e^{-tL_e})\le m_e(U_e)^2\cdot c(t/4)^4\qquad t>0
\]
where $c(t):=\|e^{-tL_e}\|_{L^2\to L^\infty}$. The stated bound follows from previous lemma.
\end{proof}
\begin{theorem}
Suppose $(V-e)_-\in L^1(\O,dx)\cap L^p(\O,dx)$ for some $p>n/2$, $n\ge 3$. Then the following bound holds true
\[
N((\D_{U_e},H^1_0(U_e,m_e)),\lambda)\le {\rm e}^{2d}S_n^d\cdot \|(V-e)_-\|_{L^1(U_e,dx)}^2\cdot \|(V-e)_-\|^{d-2}_{L^p(U_e,dx)}\cdot\lambda^{d}\qquad \lambda\ge 0\, .
\]
In particular, if $U_e=\{V<e\}$ has finite Lebesgue measure, we have
\[
N((\D_{U_e},H^1_0(U_e,m_e)),\lambda)\le {\rm e}^{2d}S_n^d\cdot |U_e|^2\cdot \|(V-e)_-\|^d_{L^p(\O,dx)}\cdot\lambda^{d}\qquad \lambda\ge 0\, .
\]
\end{theorem}
\begin{proof}
Since $\chi_{(-\infty,\lambda]}(x)\le e^{-t(x-\lambda)}$ for all $x\in\mathbb{R}$, $\lambda\ge 0$ and $t>0$, we have
\[
\begin{split}
N((\D_{U_e},H^1_0(U_e,m_e)),\lambda) &= {\rm Tr}(\chi_{(-\infty,\lambda]}(L_e)) \le {\rm Tr}(e^{-t(L_e-\lambda)}) \\
&\le \|(V-e)_-\|_{L^1(\O,dx)}^2\cdot \Bigl(edS_r(U_e,m_e)\Bigr)^{d}\cdot e^{t\lambda}\cdot t^{-d}\, .
\end{split}
\]
Choosing $t=d/\lambda$ we obtain
\[
N((D_{U_e},H^1_0(U_e,m_e)),\lambda)\le  \|(V-e)_-\|_{L^1(\O,dx)}^2\cdot \Bigl(e^2S_r(U_e,m_e)\Bigr)^{d}\cdot \lambda^{d}\qquad \lambda\ge 0\, .
\]
By Lemma 5.1 and, in particular, from the evaluation
\[
S_r(U_e,m_e)\le S_{n^*}(U_e,dx)\cdot \|(V-e)_-\|^{2/r}_{L^p(\O,dx)}\, ,
\]
we have
\[
N((D_{U_e},H^1_0(U_e,m_e)),\lambda)\le  \|(V-e)_-\|_{L^1(\O,dx)}^2\cdot \Bigl({\rm e}^2S_n\cdot \|(V-e)_-\|^{2/r}_{L^p(\O,dx)}\Bigr)^{d}\cdot \lambda^{d}\qquad \lambda\ge 0
\]
which provides the first stated bound since $r=\frac{2d}{d-2}$ implies $\frac{2d}{r}=d-2$. The second one follows from H\"older inequality.
\end{proof}

\subsection{}
To bound above $N((\E_0,\F),\lambda\|A_\lambda\|)$ we prove Sobolev inequalities for the absorption-to-reflection Dirichlet form $(\E_0,\F)$ on $L^2(\partial U_e,\mu_e)$,
\vskip0.2truecm
{\centerline {\it assuming that $U_e\subset\R^n$, $n\ge 3$, is bounded and its boundary $\partial U_e$ is smooth.}
}
\vskip0.2truecm\noindent
On $\partial U_e$ let $\sigma$ be the $(n-1)$-dimensional Hausdorff measure and consider also the measure
\[
\nu_e:=\int_{U_e}dx\, \mu_x\, .
\]
The harmonic measures and the measure $\nu_e$ depend upon the potential $V$ only through the open set $U_e:=\{V<e\}$.
\begin{lemma}
The following boundary Sobolev inequality
\begin{equation}
\|\varphi\|^2_{L^q(\partial U_e,\sigma)}\le S\cdot\E_0[\varphi]+b\cdot \|\varphi\|^2_{L^2(\partial U_e,\nu_e)}\qquad \varphi\in\F
\end{equation}
holds true for some $b\in\R$, $q:=\frac{2(n-1)}{n-2}$ and $S:=\frac{2}{n-2}\cdot \o_n^{\frac{1}{1-n}}$, where $\o_n$ is the volume of the unit sphere in $\R^n$.
\end{lemma}
\begin{proof}
Let us consider the following Sobolev trace inequality [12 Thm 0.1]
\[
\|{\rm Tr}(u)\|^2_{L^q(\partial U_e,\sigma)}\le S\cdot\D_{U_e}[u]+b\cdot \|u\|^2_{L^2(U_e,dx)}\qquad u\in H^1(U_e,dx)\, .
\]
Choosing $u=L_0\varphi$ we have ${\rm Tr}(L_0\varphi)=\varphi$ and, since by definition $\D_{U_e}[L_0\varphi]=\E_0[\varphi]$, we get
\[
\|\varphi\|^2_{L^q(\partial U_e,\sigma)}\le S\cdot\E_0[\varphi]+b\cdot \|L_0\varphi\|^2_{L^2(U_e,dx)}\qquad \varphi\in\F\, .
\]
Since moreover
\[
\begin{split}
\|L_0\varphi\|^2_{L^2(U_e,dx)}&=\int_{U_e}dx\Bigl|\int_{\partial U_e}\varphi(y)\cdot\mu_x(dy)\Bigr|^2 \le
\int_{U_e}dx \int_{\partial U_e}|\varphi(y)|^2\cdot\mu_x(dy) \\
&= \|\varphi\|^2_{L^2(\partial U_e,\nu_e)}\, ,
\end{split}
\]
we obtain the stated inequalities.
\end{proof}
\begin{lemma}
The Radon-Nikodym derivative of the measure $\nu_e$ with respect to the Hausdorff measure $\sigma$ is a continuous, nowhere vanishing function on $\partial U_e$ and there exists a constant $c_P>0$ such that
\[
\Bigl\|\frac{d\nu_e}{d\sigma}\Bigr\|_{C(\partial U_e)}     \le c_P\cdot \o_{n-1}\cdot {\rm diam}(U_e)\, .
\]
\end{lemma}
\begin{proof}
Notice that the harmonic measures $\mu_x$ on $\partial U_e$ are absolutely continuous with respect to the Hausdorff measure $\sigma$ and that their Radon-Nikodym derivatives are represented by the Poisson kernel $h:U_e\times \partial U_e\to [0,+\infty)$: $\mu_x=h(x,\cdot)\sigma$. Hence, also the measure $\nu_e$ is absolutely continuous with respect to $\sigma$, with Radon-Nikodym derivative given by
\[
\nu_e (dy)=\Bigl(\int_{U_e}dx\, h(x,y)\Bigr)\cdot\sigma(dy)
\]
Notice also that, since the function $U_e\ni x\mapsto h(x,y)$ is harmonic and nonnegative, it cannot vanish otherwise would be identically zero, by the Maximum Principle. In particular, $h(x,y)>0$ for all $(x,y)\in U_e\times \partial U_e$ and then $\int_{U_e}dx\, h(x,y)>0$ for all $y\in\partial U_e$.
\par\noindent
By a result due to Krantz [10], there exists a constant $c_P>0$ such that the Poisson kernel $h_U:U\times\partial U\to [0,+\infty)$ of a bounded, smooth domain $U\subset\R^n$ is bounded by
\[
\frac{1}{c_P}\cdot \frac{d(x,\partial U)}{|x-y|^n}\le h_U(x,y)\le c_P\cdot \frac{d(x,\partial U)}{|x-y|^n}\qquad (x,y)\in U\times \partial U\, .
\]
Applying the result to $U_e$ and since $d(x,\partial U_e)\le |x-y|$ for all $y\in\partial U_e$, we have
\[
\begin{split}
\int_{U_e}dx\, h(x,y)&\le c_P\cdot\int_{U_e}dx\, |x-y|^{1-n}\le c_P\cdot \int_{B(y,{\rm diam}(U_e))}dx\, |x-y|^{1-n} \\
& \le c_P\cdot  \int_{B(0,{\rm diam}(U_e))}dz\, |z|^{1-n} =c_P\cdot \o_{n-1}\cdot \int_0^{{\rm diam}(U_e)}dr\, r^{n-1}\cdot r^{1-n} \\
&=c_P\cdot \o_{n-1}\cdot {\rm diam}(U_e)\, .
\end{split}
\]
By a similar calculation, for $\gamma\in (1,\frac{n}{n-1})$ we have
\[
\begin{split}
\int_{U_e}dx\, h(x,y)^\gamma &\le c_P\cdot \o_{n-1}\cdot \int_0^{{\rm diam}(U_e)}dr\, r^{n-1}\cdot r^{\gamma(1-n)} \\
&=c_P\cdot \frac{\o_{n-1}\cdot {\rm diam}(U_e)^{n-\gamma(n-1)}}{n-\gamma(n-1)} <+\infty
\end{split}
\]
so that the family $\{h(\cdot,y)\in L^1(U_e, dx):y\in \partial U_e\}$ is bounded in $L^\gamma(U_e,dx)$. Consequently, by the de la Vall\'ee Poussin test, it is uniformly bounded in $L^1(U_e,dx)$. Since, moreover, for any fixed $x\in U_e$, the function $h(x,\cdot)$ is continuous on $\partial U_e$, applying the Vitali convergence Theorem we have
\[
\lim_{z\to y}\int_{U_e}dx\, h(x,z)=\int_{U_e}dx\, h(x,y)\qquad y\in\partial U_e\, .
\]
The Radon-Nikodym derivative $\frac{d\nu_e}{d\mu_e}$ is then a continuous function on the boundary $\partial U_e$.
\end{proof}
It is not clear how to bound above $c_P$ geometrically.
\begin{lemma}
Suppose $(V-e)_-\in L^p(\O,dx)$ for some $p>n/2$, $n\ge 3$. The Radon-Nikodym derivative of $\mu_e$ with respect to $\sigma$ is in $L^p(\partial U_e,\sigma)$ and
\[
\Bigl\|\frac{d\mu_e}{d\sigma}\Bigr\|_{L^p(\partial U_e,\sigma)}\le\Bigl\|\frac{d\nu_e}{d\sigma}\Bigr\|_{C(\partial U_e)}^{1-\frac{1}{p}}\cdot\|(V-e)_-\|_{L^p(U_e,dx)}\, .
\]
\end{lemma}
\begin{proof}
Setting $W:=(V-e)_-$, H\"older inequality we have
\[
\begin{split}
\int_{\partial U_e}\sigma(dy)\Bigl|\frac{d\mu_e}{d\sigma}(y)\Bigr|^p &=\int_{\partial U_e}\sigma(dy)\Bigl|\int_{U_e}dx\,  h(x,y)W(x)\Bigr|^p \\
&=\int_{\partial U_e}\sigma(dy)\Bigl|\Bigl(\int_{U_e}dx\, h(x,y)\Bigr)\cdot \Bigl(\int_{U_e}dx\, h(x,y)\Bigr)^{-1}\cdot\int_{U_e}dx\,  h(x,y)W(x)\Bigr|^p \\
&=\int_{\partial U_e}\sigma(dy)\Bigl(\int_{U_e}dx\, h(x,y)\Bigr)^p\Bigl| \Bigl(\int_{U_e}dx\, h(x,y)\Bigr)^{-1}\cdot\int_{U_e}dx\,  h(x,y)W(x)\Bigr|^p \\
&\le\int_{\partial U_e}\sigma(dy)\Bigl(\int_{U_e}dx\, h(x,y)\Bigr)^{p-1}\cdot\int_{U_e}dx\,  h(x,y)W(x)^p \\
&\le \bigl(c\cdot \o_{n-1}\cdot {\rm diam}(U_e)\bigr)^{p-1}\cdot \int_{U_e}dx\,W(x)^p \int_{\partial U_e}\sigma(dy) h(x,y)\\
&= \Bigl\|\frac{d\nu_e}{d\sigma}\Bigr\|_{C(\partial U_e)}^{p-1}\cdot \int_{U_e}dx\,W(x)^p L_0(1)(x)\\
&=\Bigl\|\frac{d\nu_e}{d\sigma}\Bigr\|_{C(\partial U_e)}^{p-1}\cdot \int_{U_e}dx\,W(x)^p \\
&= \Bigl\|\frac{d\nu_e}{d\sigma}\Bigr\|_{C(\partial U_e)}^{p-1}\cdot \|W\|^p_{L^p(U_e,dx)}\, . \\
\end{split}
\]
\end{proof}
\begin{lemma}
The Radon-Nikodym derivative of $\nu_e$ with respect to $\mu_e$ is essentially uniformly bounded with respect to the Hausdorff measure $\sigma$
\end{lemma}
\begin{proof}
Setting $W:=(V-e)_-$ we have
\[
\frac{d\mu_e}{d\nu_e}(y)=\frac{\int_{U_e} dx\, h(x,y)W(x)}{\int_{U_e} dx\, h(x,y)}\qquad y\in\partial U_e\, .
\]
Consider now the function $g:\partial U_e\to \R$ defined by
\[
g(y):=\frac{\int_{U_e} dx\, h(x,y)(W\wedge 1)(x)}{\int_{U_e} dx\, h(x,y)}\qquad y\in\partial U_e
\]
so that $\frac{d\mu_e}{d\nu_e}(y)\ge g(y)$ for all $y\in\partial U_e$.\\
By a previous lemma, the function $\partial U_e\ni y\mapsto \int_{U_e}dx\,h(x,y)$ is continuous. Since the function $W\wedge 1$ is bounded, one may prove by the same method, that the function $\partial U_e\ni y\mapsto \int_{U_e}dx\,h(x,y)(W\wedge 1)(x)$ is continuous too. Thus $g$ is a nonnegative, continuous function on the compact set $\partial U_e$ attaining its minimum value $g(y_0)$ at some point $y_0\in\partial U_e$. Since, however, the value $g(y_0)$ is, by definition, the mean value of the nonnegative function $W\wedge 1$ with respect to the finite measure $h(x,y_0)dx$ on $U_e$, it cannot vanish unless $W=(V-e)_-=0$ $dx$-almost everywhere $x\in U_e$. This is a contradiction since our running hypothesis is that $U_e:=\{V<e\}$ is a nonempty,  open set. This shows that
\[
\frac{d\nu_e}{d\mu_e}(y)\le \frac{1}{g(y_0)}<+\infty\qquad \sigma{-\rm a.e.}\,\,\ y\in\partial U_e\, .
\]
\end{proof}
\begin{lemma}
Suppose $(V-e)_-\in L^1(\O,dx)\cap L^p(\O,dx)$ for some $p>n/2$, $n\ge 3$
and set $s\in [1,q)$ such that $\frac{1}{p}+\frac{s}{q}=1$.
Then the following Sobolev inequality holds true
\begin{equation}
\|\varphi\|^2_{L^s(\partial U_e,\mu_e)}\le c_1\cdot\E_0[\varphi]+c_2\cdot\|\varphi\|^2_{L^2(\partial U_e,\mu_e)}\qquad \varphi\in\F
\end{equation}
where $c_1:=S\cdot \|\frac{d\mu_e}{d\sigma}\|^{2/s}_{L^p(\partial U_e,\sigma)}$ and $c_2:=b\cdot \|\frac{d\mu_e}{d\sigma}\|^{2/s}_{L^p(\partial U_e,\sigma)}\cdot \|\frac{d\nu_e}{d\mu_e}\|_{C(\partial U_e)}$.
\end{lemma}
\begin{proof}
In the following we shall denote by $k_e:\partial U_e\to\R$ the Radon-Nikodym derivative of $\mu_e$with respect to $\sigma$.
By H\"older inequality we have
\[
\begin{split}
\|\varphi\|^s_{L^s(U_e,\mu_e)}&=\int_{\partial U_e} \mu_e(dy)\,|\varphi(y)|^s =\int_{U_e}\, m_e(dx)\int_{\partial U_e}\, \mu_x(dy)|\varphi(y)|^s \\
&= \int_{\partial U_e}\, \sigma(dy)|\varphi(y)|^s\cdot \Bigl(\int_{U_e}m_e(dx)h(x,y)\Bigr) \\
&= \int_{\partial U_e}\, \sigma(dy)|\varphi(y)|^s\cdot k_e(y) \\
&\le \Bigl(\int_{\partial U_e} \sigma(dy)\, |\varphi (y)|^{s\frac{q}{s}}\Bigr)^{\frac{s}{q}}\cdot \Bigl(\int_{\partial U_e}\sigma(dy)|k_e(y)|^p\Bigr)^{\frac{1}{p}} \\
&=\|\varphi\|^s_{L^q(\partial U_e,\sigma)}\cdot \|k_e\|_{L^p(\partial U_e,\sigma)} \\
&\le \|k_e\|_{L^p(\partial U_e,\sigma)}\cdot \Bigl(S\cdot\E_0[\varphi]+b\cdot \|\varphi\|^2_{L^2(\partial U_e,\nu_e)}\Bigr)^{\frac{r}{2}} \\
&\le \|k_e\|_{L^p(\partial U_e,\sigma)}\cdot \Bigl(S\cdot\E_0[\varphi]+b\cdot \|\frac{d\nu_e}{d\mu_e}\|_{L^\infty(\partial U_e,\sigma)}\cdot\|\varphi\|^2_{L^2(\partial U_e,\mu_e)}\Bigr)^{\frac{s}{2}}\, .
\end{split}
\]
\end{proof}
\begin{theorem}
Suppose $(V-e)_-\in L^1(\O,dx)\cap L^p(\O,dx)$ for some $p>n/2$, $n\ge 3$
and set $r\in [1,q)$ such that $\frac{1}{p}+\frac{r}{q}=1$. Choose $r\in (2,q]$ such that $\frac{1}{p}+\frac{r}{q}=1$, where $q:=\frac{2(n-1)}{n-2}$.\\
i) Then, the Markovian semigroup $e^{-tB_0}$ is ultracontractive on $L^2(\partial U_e,\mu_e)$ and
\begin{equation}
\|e^{-tB_0}\|_{L^2\to L^\infty}\le \Bigl({\rm e}\frac{m}{4}c_1\Bigr)^{m/4}\cdot {\rm e}^{\frac{c_2}{c_1}t}\cdot t^{-m/4}\qquad t>0\, ,
\end{equation}
where $m:=\frac{2r}{r-2}$.\\
ii) If moreover $\|(V-e)_-\|_{L^1(\O,dx)}<+\infty$, then the semigroup is nuclear with
\begin{equation}
{\rm Tr}(e^{-tB_0})\le \bigl({\rm e}mc_1\bigr)^m\cdot \|(V-e)_-\|^2_{L^1(\O,dx)}\cdot {\rm e}^{\frac{c_2}{c_1}t}\cdot t^{-m}\qquad t>0\, ,
\end{equation}
iii) and the following bound hold true
\begin{equation}
N((\E_0,\F),\gamma)\le {\rm e}^{2m}\cdot \|(V-e)_-\|^2_{L^1(\O,dx)}\cdot(c_1\gamma +c_2)^m\qquad \gamma\ge 0\, .
\end{equation}
\end{theorem}
\begin{proof}
The proof of the first statement i) follows from the Sobolev inequalities (5.2) above along exactly the same lines of the proof of Lemma 5.2. Since, by hypothesis,
\[
\begin{split}
\mu_e(\partial U_e)&=\int_{U_e}m_e(dx)\int_{\partial U_e} 1\cdot \mu_x(dy)=\int_{U_e}1\cdot m_e(dx)=\|(V-e)_-\|_{L^1(\O,dx)}
\end{split}
\]
is finite, the second statement ii) follows from i) by an application of [5 Thm 2.1.4]. Finally, the last statement iii) follows from ii) optimizing over $t>0$ the bound
\[
N((\E_0,\F),\gamma)\le {\rm e}^{\gamma t}\cdot {\rm Tr}({\rm e}^{-tB_0})\qquad t>0\, .
\]
\end{proof}

%
%


\newpage
\normalsize
\begin{center} \bf REFERENCES\end{center}

\normalsize
\begin{enumerate}

\bibitem[1]{1} W. Arendt, A.F.M. ter Elst, \newblock{Ultracontractivity and eigenvalues: Weyl's law for the Dirichlet-to-Neumann operator},
\newblock{\it   Integral Equations Operator Theory} {\bf 88} {\rm (2017)}, no. 1, 65--89.

\bibitem[2]{2} W. Arendt, R. Mazzeo, \newblock{Spectral properties of the Dirichlet-to-Neumann operator on Lipschitz domains},
\newblock{\it  Ulmer Seminare, Heft 12} {\bf 88} {\rm (2007)}, 28--38.

\bibitem[3]{3} M.S. Birman, \newblock{On the spectrum of singular boundary-value problems},
\newblock{\it  Mat. Sb. (N.S.)} {\bf 55(97)} {\rm (1961)}, no. 2, 125--174.

\bibitem[4]{4} Z.-Q. Chen, M. Fukushima, \newblock{``Symmetric Markov processes, time change, and boundary theory''},
 \newblock{London Mathematical Society Monographs Series {\bf 35}, Princeton University Press, Princeton, NJ, 2012.}.

\bibitem[5]{5} H.L. Cycon, R.G. Froese, W. Kirsch, B. Simon, \newblock{``Schr\"odinger operators with application
to quantum mechanics and global geometry''},
\newblock{Texts and Monographs in Physics. Springer Study Edition. Springer-Verlag, Berlin, 1987.}

\bibitem[6]{6} E.B. Davies, \newblock{``Heat Kernels and Spectral Theory''},
\newblock{Cambridge Tracts in Mathematics,  {\bf 143}, Cambridge University Press, 1989}.

\bibitem[7]{7} E.B. Davies, B. Simon, \newblock{Ultracontractivity and the heat kernel for Schrödinger operators and Dirichlet Laplacians},
\newblock{\it  J. Funct. Anal.} {\bf 59} {\rm (1984)}, no. 2, 335--395.

 \bibitem[8]{8} A.F.M. ter Elst, E.M. Ouhabaz, \newblock{Analysis of the heat kernel of the Dirichlet-to-Neumann operator},
\newblock{\it   J. Funct. Anal.} {\bf 267} {\rm (2014)}, no. 11, 4066--4109.

\bibitem[9]{9} L. Friedlander, \newblock{Some inequalities between Dirichlet and Neumann eigenvalues},
\newblock{\it Duke Math. J.} {\bf 116} {\rm (1991)}, no. 2, 153--160.

\bibitem[10]{10} S.G. Krantz, \newblock{Calculation and estimation of the Poisson kernel},
\newblock{\it  Math. Anal. and Appl.} {\bf 302} {\rm (2005)}, 143--148.

\bibitem[11]{11} P. Li, S.-T. Yau, \newblock{On the Schr\"odinger Equation and the Eigenvalue Problem},
\newblock{\it Comm. Math. Psys.} {\bf 88} {\rm (1983)}, 309--318.

\bibitem[12]{12} Y. Li, M. Zhu, \newblock{Sharp Sobolev Trace Inequalities on Riemannian Manifolds with Boundaries},
\newblock{\it Comm Pure Appl. Math.} {\bf 50} {\rm (1997)}, no. 5, 449--487.

\bibitem[13]{13} E.H. Lieb, \newblock{``The number of bound states of one-body Schroedinger operators and the Weyl problem''},
\newblock{\it in Geometry of the Laplace operator (Proc. Sympos. Pure Math., Univ. Hawaii, Honolulu, Hawaii, 1979), Proc. Sympos. Pure Math., XXXVI, Amer. Math. Soc., Providence, R.I.}{\rm (1980)}, 241--252.

\bibitem[14]{14} L.E. Payne, \newblock{Inequalities for eigenvalues of plates and membranes},
\newblock{\it J. Rational Mech. Anal.} {\bf 4} {\rm (1955)}, 517--529.

\bibitem[15]{15} G. Polya, \newblock{On the eigenvalues of a vibrating membranes},
\newblock{\it Proc. London Math. Soc.} {\bf 11} {\rm (1961)}, no. 3, 419--433.

\bibitem[16]{16} B. Simon, \newblock{Analysis with weak trace ideals and the number of bound states of Schr\"odinger operators},
\newblock{\it  Trans. Amer. Math. Soc.} {\bf 224} {\rm (1976)}, no. 2, 367--380.

\bibitem[17]{17} J. Schwinger, \newblock{On the bound states of a given potential},
\newblock{\it  Proc. Acad. Sci. U.S.A.} {\bf 47} {\rm (1961)}, 122--129.

\bibitem[18]{18} H. Weyl, \newblock{Das asymptotische Verteilungsgesetz der Eigenwerte linearer partieller Differentialgleichungen},
\newblock{\it  Math. Ann.} {\bf 71} {\rm (1912)}, 441--479.

\end{enumerate}
\end{document}